\documentclass[cmp]{svjour_no_cmp}

\usepackage{latexsym}
\usepackage{graphics}
\usepackage{hyperref, enumerate}
\usepackage{soul}
\usepackage{graphicx}
\usepackage{amsmath}
\usepackage{amssymb}
\usepackage{verbatim}
\usepackage{color}
\usepackage{cancel}
\usepackage{doi}
\usepackage{hyperref}

\def\C{\mathbb C}

\def\Tr{\mathop {\rm Tr}}

\def\be{\begin{equation}}
\def\beq{\begin{eqnarray}}
\def\ee{\end{equation}}
\def\eeq{\end{eqnarray}}
\def\eqref#1{(\ref{#1})}
\def\lra#1{\langle #1 \rangle}
\def\lrp#1{\left( #1 \right)}

\def\abs#1{\left\vert #1\right\vert}

\def\nn{\nonumber\\}
\def\be{\begin{equation}}
\def\beq{\begin{eqnarray}}
\def\beqs{\begin{eqnarray*}}
\def\ee{\end{equation}}
\def\eeq{\end{eqnarray}}
\def\eeqs{\end{eqnarray*}}

\def\bluedoi#1{{\color{blue} \doi{#1}}}

\renewcommand{\thesection}{\Roman{section}}
\renewcommand{\thesubsection}{\thesection.\arabic{subsection}}

\newtheorem{thm}{Theorem}
\newtheorem{cor}[thm]{Corollary}
\newtheorem{lem}[thm]{Lemma}
\newtheorem{prop}[thm]{Proposition}

\usepackage{graphics}
\begin{document}
\title{Reflection Positivity for Parafermions}
\author{Arthur Jaffe\inst{1,2} \and Fabio L. Pedrocchi\inst{3,2}
}                     
\institute{Harvard University, Cambridge, Massachusetts 02138, USA\\ \email{arthur{\_}jaffe@harvard.edu} \and Department of Physics, University of Basel, Basel, Switzerland \and JARA Institute for Quantum Information, RWTH Aachen University, Aachen, Germany \\ 
\email{fabio.pedrocchi@rwth-aachen.de}}

\dedication{\centerline{Dedicated to the memory of Ursula Eva Holliger-H\"anggi.}}

%

\maketitle
\begin{abstract}
We establish reflection positivity for Gibbs trace states for a class of  gauge-invariant, reflection-invariant Hamiltonians describing parafermion interactions on a lattice. We relate these results to recent work in the condensed-matter physics literature. 
\end{abstract}

\setcounter{equation}{0}  \section{Introduction}\label{sec:Introduction}
In the early 1960's, Keijiro~Yamazaki introduced a family of algebras generalizing  a Clifford algebra.\footnote{See 1) and 2) in the middle of page 193 in \S7.5 of~\cite{Yamazaki}.} These algebras are characterized by a primitive $n^{\rm th}$ root of unity  $\omega=e^{{2\pi i}/{n}}$,  and generators $c_{j}$,  where $j=1,2,\ldots,L$, with each generator of order $n$.   Alun~Morris studied these algebras and showed that for even $L$ they have an irreducible representation on a Hilbert space $\mathcal{H}$ of dimension $N=n^{L/2}$, and this is unique up to unitary equivalence~\cite{Morris-1}.   Here we consider $L$ even and $c_{j}$ unitary.   
In the physics literature,  one calls the operators  $c_{j}$  {\em a set of parafermion generators of order~$n$} (or simply ``parafermions'') if they satisfy Yamazaki's relations:
\be\label{eq:para1}
c_{j}^{n}=I\;,
\quad\text{and}\quad
c_{j}c_{j'}=\omega \,c_{j'}c_{j}\;,
	\quad\text{for}\quad\text{$j<j'$}\;.
\ee

Consequently $c_{j}^{*}=c_{j}^{n-1}$, and also $c_{j}c_{j'}=\omega^{-1} \,c_{j'}c_{j}$ for $j>j'$.   The choice $n=2$ reduces to a self-adjoint representation of a Clifford algebra;  it describes Majoranas, namely fermionic coordinates. For $n\geqslant 3$ one obtains a generic algebra of parafermionic coordinates, whose generators are not self-adjoint.    
Note that if $\{c_{j}\}$ are a set of $L$ parafermion generators of order $n$, then $\{c_{j}^{*}\}$ is another set of $L$ parafermion generators of order $n$.   

Parafermion commutation relations appeared in both the mathematics and the physics literature, long before the definitions of the algebras cited above.  J.~J.~Sylvester introduced matrices satisfying parafermion commutation relations in 1882, see \cite{Sylvester-1,Sylvester-2}.   In 1953, Herbert S. Green proposed such commutators for fields \cite{Green}.  More recent examples occur in  \cite{Hooft,Fradkin-Kadanoff}.

The relations \eqref{eq:para1} arise from studying representations of the braid group; a new discussion appears in~\cite{CO}. Generally, representations of the braid group lead to a variety of statistics and have been the focus of intense research over the last decades, see for example \cite{Froehlich-Gabbiani}.

Paul Fendley~\cite{Fendley2,Fendley} gave a parafermion representation for Rodney Baxter's clock Hamiltonian and for some related spin chains~\cite{Baxter-1,Baxter-2,Baxter-3}, and discovered matrices similar to those in \cite{Sylvester-1}; see our remarks in \S\ref{Sect:Baxter Clock}. Some further examples occur in \cite{Baxter-4,Au-Perk}.  
Recently there has been a great deal of interest in the possibility to obtain  parafermion states  in one and two-dimensional  model systems\hbox{\cite{BarkeshliQi,CASNatCom,LBRSPRX,VPRB,BMQ,MONGPRX,KLArxiv,KL2}.}

Two sets of authors have proposed a classification of topological and non-topological phases in parafermionic chains~\hbox{\cite{MotrukBergTurnerAriPollmann,BondesanQuella}.} 

\subsection{Reflection Positivity (RP)}
Konrad Osterwalder and Robert Schrader discovered  RP for  bosons and fermion fields~\cite{OS1}, after which RP became the standard way to relate statistical physics to quantum theory, and to justify inverse Wick rotation.  Variations of this property have been central in hundreds of subsequent papers on quantum theory and also on condensed-matter physics, especially in the study of ground states and phase transitions. So RP is  fundamental, and it is important to know when it holds.   

Let $A\in\mathfrak{A}_{-}$ belong to an algebra of observables localized on one side of a reflection plane; let $\vartheta(A)$ denote the reflected observable localized  on the other side of the plane. The reflection $\vartheta$  is said to have the RP-property on $\mathfrak{A}_{-}$ with respect to the expectation $\lra{\ \cdot \ }$, if always  $\lra{{A\,\vartheta(A)}} \geqslant 0$.

In this paper we show that RP applies in lattice statistical mechanical systems generated by parafermions.  The expectation that we study here is a trace defined with the Boltzmann weight $e^{-H}$ for a class of Hamiltonians specified in \S\ref{Sect:Main}. 
Our Hamiltonians are not necessarily self-adjoint. However in case the Hamiltonian is reflection symmetric, then the partition function
	\be
		\mathfrak{Z}
		= \Tr(e^{- H}) > 0
	\ee
is automatically real and positive. 
We give our  main result  in Theorem \ref{prop:reflection_positivity_para} of \S\ref{Sect:Main}, where we show that the corresponding expectations of the form
	\be
		\lra{\ \cdot \ }
		= \Tr( \ \cdot \ e^{-H})
	\ee
are RP with respect to an {\em algebra of observables} $\mathfrak{A}_{-}^{n}$ generated by monomials in parafermions of degree $n$.   This paper generalizes our earlier results on the  algebra of fermionic coordinates~\cite{JP}.

\setcounter{equation}{0}
\section{Basic Properties of Monomials in Parafermions}
Parafermions  $c_{j}$ yield ordered monomials with exponents taken mod $n$, 
\begin{equation}\label{eq:Mbeta}
C_{\mathfrak{I}}=c_{1}^{n_{1}}c_{2}^{n_{2}}\cdots c_{L}^{n_{L}}\,,
\quad\text{where}\quad
0\leqslant n_{j}\leqslant n-1\;.
\end{equation}
Define the set of exponents,  
$\mathfrak{I}=\{n_{1},\ldots,n_{L}\}$, and denote the total degree as   
	\be
		\vert\mathfrak{I}\vert 
		=\sum_{j=1}^{L} n_{j}\;.
	\ee

\subsection{Algebras of Parafermions}
The parafermion monomials $C_{\mathfrak{I}}$ generate an algebra that we denote $\mathfrak{A}$. 
Divide  the $L$ parafermions $c_{i}$ into two subsets, according to whether  or not $i\leqslant \frac12 L$.
Define $\mathfrak{A}_{-}$ as the algebra generated by monomials $C_{\mathfrak{I}}$, for which $n_{j}=0$ for all $j>\frac{1}{2}L$. 
Correspondingly let $\mathfrak{A}_{+}$ denote the algebra generated by monomials $C_{\mathfrak{I}}$, for which $n_{j}=0$ for all $j\leqslant\frac{1}{2}L$. 
In addition, define the  ``order $k$''-parafermion subalgebras $\mathfrak{A}^{k}_{\pm}\subset \mathfrak{A}_{\pm}$ as follows: 
	\be\label{eq:Subalgebras}
		\mathfrak{A}^{k}_{\pm} \text{ is the algebra generated by }
		 C_{\mathfrak{I}}\in\mathfrak{A}_{\pm}\;,
		\text{ with } \abs{\mathfrak{I}} =k\;.
	\ee

One can add the sets indexing parafermions by setting
	\be
		\mathfrak{I} + \mathfrak{I}'
		= \{n_{1}+n'_{1}, \ldots, n_{L}+n'_{L}\}\;.
	\ee
Clearly there is no loss in generality to require that one takes each sum $n_{j}+n'_{j}$ mod $n$. Define the numbers
	\be\label{eq:Contractions}
		\mathfrak{I}\circ \mathfrak{I}'
		= \sum_{1\leqslant j<j'\leqslant L} n_{j} n'_{j'}\;,
		\quad\text{and}\quad
		\mathfrak{I}\wedge \mathfrak{I}'
		= \mathfrak{I}\circ \mathfrak{I}'
		- \mathfrak{I}'\circ \mathfrak{I}\;.
	\ee
With these definitions   
	\be\label{Reverse C}
		C_{\mathfrak{I}}C_{\mathfrak{I}'}
		= \omega^{-\mathfrak{I}\circ \mathfrak{I}'}\,
		C_{\mathfrak{I}+ \mathfrak{I}'}
		= \omega^{-\mathfrak{I}\wedge \mathfrak{I}'}\,C_{\mathfrak{I}'}C_{\mathfrak{I}}\;.
	\ee 		
Denote  the complement of $\mathfrak{I}$ by 
	$
	 	\mathfrak{I}^{c}=\{n-n_{1},\ldots,n-n_{L}\}		
	$.
One has 
	\be\label{eq:Adjoint C}
		C_{\mathfrak{I}}^{*}
		= 
		\omega^{-\mathfrak{I}\circ\mathfrak{I}}\,
		C_{\mathfrak{I}^{c}}\;,
		\quad\text{and}\quad
		C_{\mathfrak{I}}^{*} \,C_{\mathfrak{I}}=I
		= C_{\mathfrak{I}}\,C_{\mathfrak{I}}^{*}\;.
	\ee

\subsection{Reflection}
Define the reflection $\vartheta$ as the map 
	\be
		i\mapsto \vartheta i =  L-i  +1\;.
	\ee 
Represent $\vartheta $ as an anti-unitary operator on $\mathcal{H}$. Conjugation by $\vartheta$ (which we denote $\vartheta(A)$) yields an anti-linear automorphism of the algebra $\mathfrak{A}$, 
	\be\label{Reflection Definition}
		\vartheta (c_{i}) 
		= \vartheta c_{i} \vartheta^{-1}
		=c_{\vartheta i}^{*}=c_{\vartheta i}^{n-1}\;,
		\quad\text{and}\quad
		\vartheta\lrp{c_{j}c_{k}}=\vartheta\lrp{c_{j}}\vartheta\lrp{c_{k}}\,.
	\ee

Set  $\vartheta\mathfrak{I}
		=\{n_{L}, \ldots, n_{1}\}$, and note that $\lrp{\vartheta\mathfrak{I}}^{c}
		= \vartheta(\mathfrak{I}^{c})=\vartheta\mathfrak{I}^{c}$.  Using \eqref{eq:Contractions}, one sees that 
	\be\label{Reflection C}
		\vartheta(C_{\mathfrak{I}})
		= \omega^{-\mathfrak{I}\circ\mathfrak{I}}  \,C_{\vartheta\mathfrak{I}^{c}}\;.
	\ee
Take $\Lambda_{-}=\{1, 2, \ldots, L/2\}$ and  $\Lambda_{+}=\{L/2+1,  \ldots, L\}$ to divide the points $\Lambda=\Lambda_{-}\cup\Lambda_{+}$ into two sets  $\Lambda_{\pm}$  exchanged by reflection. To simplify notation,  we relabel the sites in order to put sites $1$ to $L/2$ on one side of the reflection plane and sites $L/2+1$ to $L$ on the other side.  Periodic boundary conditions would relate sites $1$ and $L$.

By definition $\mathfrak{A}$ is the algebra generated by the parafermions $c_{j}$ with $j\in\Lambda$. 
Denote $C_{\mathfrak{I}}\subset\mathfrak{A}_{\pm}$ also by  
$\mathfrak{I}\subset\Lambda_{\pm}$.  In this case $n_{j}=0$ for all $j>L/2$.
For $\mathfrak{I}\subset\Lambda_{+}$ and $\mathfrak{I'}\subset \Lambda_{-}$, one has $\mathfrak{I} \circ \mathfrak{I}'=0$.  So in this case   
	\be\label{-Wedge+}
		\mathfrak{I} \wedge \mathfrak{I}'
		= -\mathfrak{I}' \circ \mathfrak{I}
		= - \sum_{j,j'} n_{j}  n^{\prime}_{j^{\prime}} 
		= - \abs{\mathfrak{I}^{\phantom\prime}}
		 \abs{\mathfrak{I}^{\prime}}\;.
	\ee

\subsection{Gauge Transformations}	
We introduce the family of local gauge automorphisms $U_{j}$ defined by 
	\be
		c_{j} \mapsto U_{j'}(c_{j}) = \omega^{\delta_{jj'}} \,c_{j}\;,
		\quad\text{for}\quad
		j=1,\ldots, L\;.
	\ee
Here $\delta_{jj'}$ is the Kronecker delta function.   As shown in \cite{CO}, this transformation can be implemented on the Hilbert space of parafermions by the unitary transformation $V_{j}=e^{-2\pi iN_{j}/n}$, where $N_{j}$ is a parafermionic number operator, and $U_{j'}(c_{j})=V_{j'}c_{j}V_{j'}^{*}$.  The different $V_{j}$ commute.

Global gauge transformations are defined by $U=\prod_{j=1}^{L} U_{j}$ and  transform all parafermions by the same phase $\omega$. 
Special significance is attached to the parafermion monomials that are invariant under global gauge transformations.  In fact we say that  the globally-gauge-invariant parafermion monomials are {\em observables}. We call the gauge-invariant algebra $\mathfrak{A}^{n}$ the {\em algebra of observables}.

\setcounter{equation}{0}
\section{Reflection Symmetry and Gauge Invariance}
Here we show that certain multiples of the monomials \eqref{eq:Mbeta} are both reflection-symmetric and gauge invariant.  These monomials may not be hermitian.    We also discuss the general form of reflection-symmetric, gauge-invariant, polynomial Hamiltonians. 

\begin{lem}[Elementary Rearrangement]  \label{Lemma:Elementary Rearrangement}
For $\mathfrak{I}_{\pm}\subset \Lambda_{\pm}$,	
	\be\label{eq:Reverse Theta C C}
		C_{\mathfrak{I}_{+}}\, C_{\mathfrak{I}_{-}}
		= \omega^{-\abs{\mathfrak{I}_{+}}\abs{\mathfrak{I}_{-}}}\,C_{\mathfrak{I}_{-}}\,C_{\mathfrak{I}_{+}}\;.
	\ee
Also for $\mathfrak{I},\mathfrak{I}'\subset\Lambda_{-}$, 
	\be\label{eq:Reverse Theta C C-2}
\vartheta(C_{\mathfrak{I}})\, C_{\mathfrak{I}'}
		= \omega^{\abs{\mathfrak{I}\phantom{'}}\,
		\abs{\mathfrak{I}'}}\,
		C_{\mathfrak{I}'}  \,\vartheta(C_{\mathfrak{I}})\;.
	\ee
\end{lem} 

\begin{proof}
For $\mathfrak{I}_{\pm}\subset\Lambda_{\pm}$, one has $\mathfrak{I}_{+}\circ\mathfrak{I}_{-}=0$.  Hence
	\be\label{Wedge Vanishing}
		\mathfrak{I}_{+}\wedge\mathfrak{I}_{-}
		=-\mathfrak{I}_{-}\circ\mathfrak{I}_{+}
		=-\abs{\mathfrak{I}_{-}} \,\abs{\mathfrak{I}_{+}}\;.
	\ee  
Therefore \eqref{Reverse C} can be written in this case as 
\eqref{eq:Reverse Theta C C}.  Also $\vartheta\mathfrak{I}^{c}\in\Lambda_{+}$, so \eqref{Reflection C} and \eqref{Wedge Vanishing} ensures   
	\be
		\vartheta(C_{\mathfrak{I}})\, C_{\mathfrak{I}'}
		= \omega^{-\vartheta\mathfrak{I}^{c}\circ 
		{\mathfrak{I}'}}
		\,
		C_{\mathfrak{I}'}\,\vartheta(C_{\mathfrak{I}}) \;.
	\ee
But $\abs{\vartheta\mathfrak{I}^{c}}=nL-\abs{\mathfrak{I}}$,  so  \eqref{eq:Reverse Theta C C-2} holds.  \hfill $\square$
\end{proof}

\begin{prop}\label{Prop:RS Phase Sigma}
Let $C_{\mathfrak{I}}\in\mathfrak{A}_{-}$ have the form \eqref{eq:Mbeta},  and let
	\be\label{eq:X term definition}
		X_{\mathfrak{I}}=
		\omega^{\frac{1}{2}\abs{\mathfrak{I}}^{2}}\,C_{\mathfrak{I}}\,
		\vartheta(C_{\mathfrak{I}})\;,
		\quad\text{where} \quad
		\omega=e^{\frac{2\pi i}{n}}\;.
	\ee
Then $X_{\mathfrak{I}}$ is both reflection invariant and globally gauge invariant.
More generally for $X_{\mathfrak{I}}=e^{i\theta}\,C_{\mathfrak{I}}\,
		\vartheta(C_{\mathfrak{I}})$, the reflection-invariant combination $X_{\mathfrak{I}} + \vartheta(X_{\mathfrak{I}})$ is a real multiple of \eqref{eq:X term definition}.
\end{prop}

\begin{proof}
One has 
	\begin{eqnarray}\label{eq:varthetaX}
		\vartheta(X_{\mathfrak{I}})
		&=&\vartheta(\omega^{\frac{1}{2}\vert\mathfrak{I}\vert^{2}}\,
			C_{\mathfrak{I}}\,\vartheta(C_{\mathfrak{I}}))
		= {\omega^{-\frac{1}{2}\vert\mathfrak{I}\vert^{2}}}\,
			\vartheta(C_{\mathfrak{I}})\,C_{\mathfrak{I}}\;.
\eeq
Substitute the elementary rearrangement of Lemma \ref{Lemma:Elementary Rearrangement} with $\mathfrak{I}=\mathfrak{I}'$  into \eqref{eq:varthetaX}.  This entails $\vartheta(X_{\mathfrak{I}})=X_{\mathfrak{I}}$ as claimed.  

Furthermore $X_{\mathfrak{I}}$ is a globally-gauge-invariant monomial, for 
	\be
		U C_{\mathfrak{I}} U^{*}
		= \omega^{\vert\mathfrak{I}\vert}\,C_{\mathfrak{I}}\;,
		\quad\text{while}\quad
		U \vartheta(C_{\mathfrak{I}})U^{*}
		=  \omega^{-\vert\mathfrak{I}\vert}\,
		\vartheta(C_{\mathfrak{I}})\;.
	\ee
As $U$ is linear, we infer $UX_{\mathfrak{I}}U^{*}=X_{\mathfrak{I}}$.

The second assertion also follows, by noting that the multiple in question is  $2\cos\lrp{\theta - \frac12 \abs{\mathfrak{I}}^{2}}$.   \hfill $\square$
\end{proof}

\begin{cor}
Reflection-invariant, globally-gauge-invariant polynomials that are linear combinations of monomials \eqref{eq:X term definition} can be written as  
	\be\label{Interaction Term-1}
		\sum_{\mathfrak{I}\subset\Lambda_{-}   \atop  \abs{\mathfrak{I} }>0 } (-1)^{{1+ \abs{\mathfrak{I}}} }\,
		\omega^{ \frac{1}{2}{\abs{\mathfrak{I}}^{2}}}
		J_{\mathfrak{I}\,\vartheta{\mathfrak{I}}} \,
		C_{\mathfrak{I}}\,\vartheta(C_{\mathfrak{I}})\;,
		\quad\text{with real couplings}\ 
		J_{\mathfrak{I}\,\vartheta{\mathfrak{I}}}\;.
	\ee
\end{cor}

\subsection{Hermitian Hamiltonians}
In general a monomial $Y_{\mathfrak{I}}$ entering the sum \eqref{Interaction Term-1} is not hermitian, but 
	\beq
		Y_{\mathfrak{I}}^{*}
		&=& \omega^{-\abs{\mathfrak{I}}^{2}}
		(-1)^{{1+ \abs{\mathfrak{I}}} }\,
		\omega^{\frac12\abs{\mathfrak{I}}^{2}}
		J_{\mathfrak{I}\,\vartheta{\mathfrak{I}}} \,
		\vartheta(C_{\mathfrak{I}}^{*})\, C_{\mathfrak{I}}^{*}\nn
		&=& (-1)^{{1+ \abs{\mathfrak{I}}} }\,
		\omega^{\frac12\abs{\mathfrak{I}}^{2}}
		J_{\mathfrak{I}\,\vartheta{\mathfrak{I}}} \,
		C_{\mathfrak{I}^{c}}\,
		\vartheta(C_{\mathfrak{I}^{c}})\;.
	\eeq
In the second equality we use \eqref{eq:Adjoint C} and \eqref{eq:Reverse Theta C C}.  Therefore, the monomial $Y_{\mathfrak{I}}$ is hermitian only if 
$\mathfrak{I}^{c}=\mathfrak{I}$.  This entails $n_{i}=\frac12 n$, for every $i$.  So a necessary condition for $Y_{\mathfrak{I}}$  to be hermitian  is that $n$ is even.

For example if  $n=2$ and $L=2$ with the two sites $\vartheta 1=2$, one can take $\vert\mathfrak{I}\vert=1$.  Then $\omega=-1$, and  the monomial 
	\be
		Y_{\mathfrak{I}} 
		= i J_{\mathfrak{I}\,\vartheta\mathfrak{I}} \,c_{1}\vartheta(c_{1}) 
	\ee
has the form \eqref{Interaction Term-1}; it 
is both reflection-symmetric and hermitian.  On the other hand, any such $Y_{\mathfrak{I}}$  yields the polynomial  $Y_{\mathfrak{I}} + Y_{\mathfrak{I}}^{*}$, that is  both reflection symmetric and  hermitian. For example, with $n=3$ and $L=2$, one has $\omega=e^{\frac{2\pi i}{3}}$.  Then the monomial 
	\be
		Y_{\mathfrak{I}}
		=\omega^{\frac12} J_{\mathfrak{I} \,\vartheta{\mathfrak{I}}}\,c_{1} \vartheta(c_{1})
		=\omega^{\frac12} J_{\mathfrak{I} \,\vartheta{\mathfrak{I}}}\,c_{1} c_{2}^{*}
		=\omega^{\frac12} J_{\mathfrak{I} \,\vartheta{\mathfrak{I}}}\,c_{1} c_{2}^{2}\;,
	\ee
yields the reflection-symmetric, hermitian polynomial   
\be
Y_{\mathfrak{I}} + Y_{\mathfrak{I}}^{*}
= \omega^{\frac{1}{2}}J_{\mathfrak{I}\, \vartheta\mathfrak{I}} \,(c_{1}c_{2}^{2}+c_{1}^{2}c_{2})\;.
\ee

\setcounter{equation}{0}
\section{A Basis for Parafermions}\label{Sect:Basis}
Let $C_{\mathfrak{I}}=c_{1}^{n_{1}}\cdots c_{L}^{n_{L}}$ be one of $n^{L}$ monomials of the form \eqref{eq:Mbeta}, with $L$ even.  Let $C_{\mathfrak{I}}$ act on a Hilbert space $\mathcal{H}$ of $\dim(\mathcal {H})=n^{L/2}$.

\begin{prop}\label{prop:monomials_para}
The monomials $C_{\mathfrak{I}}$ are linearly independent, and provide a basis for the $n^{L}$ linear transformations on $\mathcal{H}$. Furthermore  $\Tr\lrp{C_{\mathfrak{I}}}=0$, unless $\vert\mathfrak{I}\vert=0$.     Any linear transformation $A$ on $\mathcal{H}$ has the decomposition 
	\be\label{eq:A_expansion}
A=\sum_{\mathfrak{I}}a_{\mathfrak{I}}\,C_{\mathfrak{I}}\,,\quad\text{where}\quad a_{\mathfrak{I}}=\frac{1}{n^{L/2}}\Tr\left(C_{\mathfrak{I}}^{*}A\right)\,.
	\ee
\end{prop}

	\begin{proof}
If $C_{\mathfrak{I}}=I$, then $\Tr\lrp{C_{\mathfrak{I}}}=\dim\mathcal{H}=n^{L/2}$.   So we need only analyze $\vert\mathfrak{I}\vert>0$.  We consider two cases.
\paragraph{Case I: A particular $c_{j}$ does not occur in $C_{\mathfrak{I}}$}

Distinguish between two subcases, according to whether or not $\sum_{i<j}n_{i}-\sum_{i>j}n_{i}=0$ mod $n$. If this quantity does not vanish, then cyclicity of the trace and the parafermion relations \eqref{eq:para1} ensure that  
	\begin{eqnarray*}
		\Tr\lrp{C_{\mathfrak{I}}}
		&=& \Tr\lrp{C_{\mathfrak{I}}c_{j}^{n}}
		= \Tr\lrp{ c_{j}C_{\mathfrak{I}}c_{j}^{n-1}}\nonumber \\
		&=& \omega^{\sum_{i<j}n_{i}-\sum_{i>j}n_{i}}\,\Tr\lrp{C_{\mathfrak{I}}}\;.
	\end{eqnarray*}
The last equality is a consequence of \eqref{eq:para1}, allowing one  to move $c_{j}$ to the right through $C_{\mathfrak{I}}$.  As $ \omega^{\sum_{i<j}n_{i}-\sum_{i>j}n_{i}}\neq1$, we infer that $\Tr(C_{\mathfrak{I}})=0$.  

On the other hand, when $\sum_{i<j}n_{i}-\sum_{i>j}n_{i}=0\mod n$, there exists  $j'\neq j$ with $n_{j'}\neq0 \mod n$, and also $\vert j-j'\vert$ is minimized.  If $j'<j$,  then
	\begin{eqnarray*}
		\Tr\lrp{C_{\mathfrak{I}}}
		&=& \Tr\lrp{C_{\mathfrak{I}}c_{j'}^{n}}
		= \Tr\lrp{ c_{j'}C_{\mathfrak{I}}c_{j'}^{n-1}}\nonumber \\
		&=& \omega^{- n_{j'}+(\sum_{i<j}n_{i}-\sum_{i>j}n_{i})}\,\Tr\lrp{C_{\mathfrak{I}}}\,\\
		&=&\omega^{- n_{j'}}\,\Tr\lrp{C_{\mathfrak{I}}}\,
		=0\;.
	\end{eqnarray*}
In the last equality we use that $\omega^{n_{j'}}\neq 1$.
If $j'>j$ the same reasoning can be followed, except  $\omega^{n_{j'}}$ replaces $\omega^{-n_{j'}}$.

\paragraph{Case II: Every $c_{j}$ occurs in $C_{\mathfrak{I}}$.}Here we have
	\be\label{C-0}
		n_{j} \in \{1, 2, \ldots,n-1\}\;,
	\ee  
for each $j$. Move  
one of the $c_{j}$'s cyclically through the trace, and back to its original position.  For $j=1$, this shows that 
	\be
		\Tr\lrp{C_{\mathfrak{I}}}
		= \omega^{\sum_{j=2}^{L }n_{j}} \Tr\lrp{C_{\mathfrak{I}}}\;.
	\ee
Hence either $\Tr\lrp{C_{\mathfrak{I}}}=0$, or else
	\be\label{C-1}
		\sum_{j=2}^{L }n_{j}  = 0\mod n\;.
	\ee
Likewise for $2\leqslant j \leqslant L$, either  $\Tr(C_{\mathfrak{I}})=0$, or 	\be\label{C-2}
		-\sum_{j=1}^{k-1}n_{j} +\sum_{j=k+1}^{L }n_{j}  = 0\mod n\;,
		\quad\text{for}\quad
		k=2, \ldots, L-1\;,
	\ee
and for $k=L$,
	\be\label{C-3}
		\sum_{j=1}^{L-1}n_{j}= 0\mod n\;.
	\ee

The conditions \eqref{C-1} and \eqref{C-2} for the case $k=2$, show that $n_{1}+n_{2}=0\mod n$. Condition \eqref{C-0} ensures that both $n_{1}$ and $n_{2}$ are strictly greater than $0$ and strictly less than $n$, so  $n_{1}+n_{2}=n$.  

Next subtract the condition \eqref{C-2} for $k=3$ from the same condition for $k=2$.   This shows that $n_{2}+n_{3}=0\mod n$, and the restriction \eqref{C-0} ensures that $n_{2}+n_{3}=n$.  Continue in this fashion for $k=j+1$ and $k=j$, in order to infer that $n_{j}+n_{j+1}=n$ for $j=3, \ldots, L-2$.  Finally consider the condition \eqref{C-3}.  As we have seen that $n_{j}+n_{j+1}=n$ for $j=1, 3, 5, \ldots, L-3$, we infer that  $n_{L-1}=0 \mod n$.  But this   is incompatible with $1<n_{L-1}<n$ required by  \eqref{C-0}.   So we conclude that  $\Tr(C_{\mathfrak{I}})=0$ in all cases for which $\mathfrak{I}\neq0$. 	

Note that $C_{\mathfrak{I}}^{*}C_{\mathfrak{I}}=I$ for each $\mathfrak{I}$.  Assuming that $\mathfrak{I}\neq\mathfrak{I}'$, it follows from the form \eqref{eq:Mbeta} for $C_{\mathfrak{I} }$, that $C_{\mathfrak{I}'}^{*}C_{\mathfrak{I}}=\pm C_{\gamma}$ for some $\gamma\neq0$.
Suppose that there are coefficients $a_{\mathfrak{I}}\in\mathbb{C}$ such that $\sum_{\mathfrak{I}} a_{\mathfrak{I}}C_{\mathfrak{I}}=0$.  Then for any $\mathfrak{I}'$, one has ${C_{\mathfrak{I}'}^{*} \sum_{\mathfrak{I}} a_{\mathfrak{I}}C_{\mathfrak{I}}}= \sum_{\mathfrak{I}} a_{\mathfrak{I}} {C_{\mathfrak{I}'}^{*}C_{\mathfrak{I}}}=0$. Taking the trace shows that $a_{\mathfrak{I}'}=0$, so the $C_{\mathfrak{I}}$ are actually linearly independent. As there are $n^{L}$ linearly independent matrices $C_{\mathfrak{I}}$, namely the square of the dimension of the representation space $n^{L/2}$ of parafermions,  these monomials are a basis set for all matrices.   Expanding an arbitrary matrix $A$ in this basis, we calculate the coefficients in \eqref{eq:A_expansion} using $\Tr I=n^{L/2}$.  \hfill $\square$
	\end{proof}

\setcounter{equation}{0}
\section{Primitive Reflection-Positivity}
\begin{prop}\label{prop:positivity_para}
Consider an operator $A\in\mathfrak{A}_{\pm}$, then
\be
{\Tr}(A\,\vartheta(A))\geqslant0\,.
\ee
\end{prop}

\begin{proof}
The operator $A\in\mathfrak{A}_{\pm}$ can be expanded as a polynomial in the basis $C_{\mathfrak{I}}$ of Proposition \ref{prop:monomials_para}. One can restrict to $\mathfrak{I}\in\Lambda_{\pm}$, so the monomials that appear in the expansion all belong to $\mathfrak{A}_{\pm}$.   Write
\be
A=\sum_{\mathfrak{I}}a_{\mathfrak{I}}\,C_{\mathfrak{I}}\;,\qquad
\text{and}\quad
\vartheta(A)=\sum_{\mathfrak{I}} \overline{a_{\mathfrak{I}}}\,\vartheta(C_{\mathfrak{I}})\;.
\ee
With $A\in\mathfrak{A}_{-}$, one  can take $C_{ \mathfrak{I}}=c_{1}^{n_{1}}\cdots c_{L/2}^{n_{L/2}}$, so 
	\be
		\Tr\lrp{A\,\vartheta(A)}
		=
		\sum_{\mathfrak{I},\mathfrak{I}^{'}}
		a_{\mathfrak{I}} \,\overline {a_{\mathfrak{I}'}}\, \Tr\lrp{C_{\mathfrak{I}}\,\vartheta(C_{\mathfrak{I}'})}
	\,.
	\ee
Since $C_{\mathfrak{I}}\in\mathfrak{A}_{-}$ and $\vartheta(C_{\mathfrak{I}'})\in\mathfrak{A}_{+}$, they are products of different parafermions. We infer from Proposition \ref{prop:monomials_para} that the trace vanishes unless $\vert\mathfrak{I}\vert=\vert\vartheta\mathfrak{I}'\vert=0$. Then 
\be
\text{Tr}\left(A\,\vartheta(A)\right)
= n^{L/2}
\left\vert a_{0}\right\vert^2\geqslant0\,,
\ee
as claimed.    \hfill $\square$
\end{proof}

\setcounter{equation}{0}
\section{The Main Results}\label{Sect:Main}
Fix the order $n$ of parafermions, and consider positive-temperature states determined by a Hamiltonian $H$ that is reflection invariant $\vartheta(H)=H$, and globally gauge invariant $UHU^{*}=H$.  But $H$ is  not necessarily  hermitian.  

Assume that $H$ has the form 
\begin{equation}\label{eq:Hamiltonian_para}
H=H_{-}+H_{0}+H_{+}\,,
\end{equation}
with $H_{\pm}\in\mathfrak{A}_{\pm}^{n}$ and $H_{+}=\vartheta(H_{-})$.
Here  $H_{0}$ is a sum of interactions \eqref{Interaction Term-1} across the reflection plane, namely 
	\begin{equation}\label{eq:H0}
		H_{0}
		=\sum_{\mathfrak{I}\subset\Lambda_{-}\atop \vert \mathfrak{I}\vert >0}(-1)^{\vert\mathfrak{I}\vert+1}\omega^{\frac{1}{2}\vert\mathfrak{I}\vert^{2}}
		\,J_{\mathfrak{I}\,\vartheta\mathfrak{I}}\,C_{\mathfrak{I}}\,			\vartheta(C_{\mathfrak{I}})\,.
	\end{equation}

\subsection{Assumptions on the coupling constants\label{Sect:Coupling Constants}} For any $n$, our results hold if the coupling constants in \eqref{eq:H0} satisfy
	\beq	\label{eq:CouplingRestriction-2}
		\phantom{(-1)^{\abs{\mathfrak{I}}}}
		  J_{\mathfrak{I}\,\vartheta \mathfrak{I}}  &\geqslant&0\,,\quad\text{for all }\mathfrak{I}\,.
	\eeq
Alternatively, for even $n$, our results hold if the coupling constants satisfy  
	\beq	\label{eq:CouplingRestriction-1}
			(-1)^{\abs{\mathfrak{I}}} J_{\mathfrak{I}\,\vartheta \mathfrak{I}}&\geqslant&0\;,\quad\text{for all } \mathfrak{I}\;.
	\eeq
	
Note that we only restrict the signs of the coupling constants for those  interactions that cross the reflection plane.\footnote{ The conditions \eqref{eq:CouplingRestriction-2}--\eqref{eq:CouplingRestriction-1}, taken together with our definition \eqref{eq:H0} for the phase of the couplings, reduce to the conditions  in our earlier work on Majoranas \cite{JP}, for which $n=2$ and $\omega=-1$.  The phase in \eqref{eq:H0} is $i^{2\abs{\mathfrak{I}}+2+ \abs{\mathfrak{I}}^{2}}=-1,i$, corresponding to $\abs{\mathfrak{I}}$ being even or odd, respectively. In \cite{JP} the corresponding phases were $i^{(\abs{\mathfrak{I}} \mod 2)}=1, i$.  Thus the couplings $J_{\mathfrak{I} \,\vartheta\mathfrak{I}}$ in the present paper have the opposite sign from those in \cite{JP} for even  $\abs{\mathfrak{I}}$; they have the same sign  for odd  $\abs{\mathfrak{I}}$.  Bearing this in mind, the allowed interactions in the two papers agree for $n=2$. For the case of general $n$, our new choice of signs simplifies the formulation of  conditions \eqref{eq:CouplingRestriction-2}--\eqref{eq:CouplingRestriction-1}.}   
The functional 
\be\label{eq:rp functional para}
\Tr(A\,\vartheta(B)\,e^{-H})\,,
\qquad\text{for}\quad
A,B\in \mathfrak{A}_{\pm}^{n}\;,
\ee
that is linear in $A$ and anti-linear in $B$ defines a pre-inner product.  

\subsection{Reflection Positivity on the Algebra of Observables}
Here we show that a reflection symmetric, globally-gauge-invariant Hamiltonian  $H$ has the reflection-positivity property on the algebra $\mathfrak{A}_{\pm}^{n}$ of gauge-invariant observables. 

\begin{thm}\label{prop:reflection_positivity_para}
Let $A\in\mathfrak{A}_{\pm}^{n}$ and $H$ of the form \eqref{eq:Hamiltonian_para}--\eqref{eq:CouplingRestriction-1}.  Then the functional \eqref{eq:rp functional para} is positive on the diagonal,
\be\label{eq:rp}
\Tr(A\,\vartheta(A)\,e^{-H})=\Tr(\vartheta(A)\,A\,e^{-H})
\geqslant0\,-
\ee
In particular, the partition function  $\Tr(e^{-H})\geqslant0$ is real and non-negative.
\end{thm}

\begin{proof} 
Use  the Lie product formula for matrices $\alpha_{1}$, $\alpha_{2}$, and $\alpha_{3}$ in the form 
\be\label{eq:LT_para}
 e^{\alpha_{1}+\alpha_{2}+\alpha_{3}}=\lim\limits_{k\rightarrow\infty}\left((1+\alpha_{1}/k)e^{\alpha_{2}/k}e^{\alpha_3/k}\right)^{k}\;,
\ee
with   $\alpha_{1}=-H_{0}$, $\alpha_{2}=-H_{-}$, and $\alpha_{3}=-H_{+}$.
(Such an approximation was also used in equation (2.6) of \cite{Froehlich-Lieb}.)   Using \eqref{eq:LT_para},   one has $e^{-H}=\lim_{k\to\infty}
\lrp{e^{-H}}_{k}$, 
where
	\beq\label{eq:Exp_minusH_k para}
		&&\hskip-.5cm \lrp{e^{-H}}_{k} \nn
		&& =
		\lrp{(I
		-\frac{1}{k}
		\,\sum_{\mathfrak{I}\subset\Lambda_{-}\atop \vert \mathfrak{I}\vert >0}(-1)^{1+\vert\mathfrak{I}\vert} \omega^{\frac{1}{2}\vert\mathfrak{I}\vert^{2}}{J_{\mathfrak{I}\,\vartheta \mathfrak{I}}}
		\,C_{\mathfrak{I}}\,\vartheta(C_{\mathfrak{I}}))
		\,e^{-H_{-}/k}\,e^{-\vartheta(H_{-})/k}}^{k}\nn
		&& =
		\lrp{(I
		+\frac{1}{k}\sum_{\mathfrak{I}\subset\Lambda_{-}\atop \vert \mathfrak{I}\vert >0}(-1)^{\vert\mathfrak{I}\vert} \omega^{\frac{1}{2}\vert\mathfrak{I}\vert^{2}}{J_{\mathfrak{I}\,\vartheta \mathfrak{I}}}
		\,C_{\mathfrak{I}}\,\vartheta(C_{\mathfrak{I}}))
		\,e^{-H_{-}/k}\,e^{-\vartheta(H_{-})/k}}^{k}\;.\nn
	\eeq
One can include the term $I$ in the sums in \eqref{eq:Exp_minusH_k para} by defining $J_{\varnothing\,\vartheta\varnothing}=k$,  
and including $\abs{\mathfrak{I}}=0$ in the sum.  Then
	\beq\label{eq:Exp_minusH_k-2}
		\lrp{e^{-H}}_{k}
		&=&
		\frac{1}{k^{k}}
		\lrp{\sum_{\mathfrak{I}\subset\Lambda_{-}}
		{(-1)^{\vert\mathfrak{I}\vert}\omega^{\frac{1}{2}\vert\mathfrak{I}\vert^{2}}J_{\mathfrak{I}\,\vartheta \mathfrak{I}}}
		\,C_{\mathfrak{I}}\,\vartheta(C_{\mathfrak{I}})
		\,e^{-H_{-}/k}\,e^{-\vartheta(H_{-})/k}}^{k}\nonumber\\
		&=& \sum_{\mathfrak{I}^{(1)},\ldots,\,\mathfrak{I}^{(k)}
		\subset\Lambda_{-}}
		(-1)^{\sum_{j=1}^{k}\vert\mathfrak{I}^{(j)}\vert}\omega^{\sum_{j=1}^{k}\frac{1}{2}\vert\mathfrak{I}^{(j)}\vert^{2}}\nn
		&& \hskip 1.2in \times
		\mathfrak{c}_{\mathfrak{I}^{(1)},\ldots,\,\mathfrak{I}^{(k)}}\,Y_{\mathfrak{I}^{(1)},\ldots,\,\mathfrak{I}^{(k)}}\;.
	\eeq
In the second equality we have expanded the expression into a linear combination of terms  with coefficients 
	\be\label{eq:product_couplings}
		\mathfrak{c}_{\mathfrak{I}^{(1)},\ldots,\,\mathfrak{I}^{(k)}}
		= \frac{ 1}{k^{k}} \prod_{j=1}^{k}J_{\mathfrak{I}^{(j)}\,			\vartheta\mathfrak{I}^{(j)}}\,,
	\ee 
and with 
	\beq\label{eq:Y_para}
		Y_{\mathfrak{I}^{(1)},\ldots,\,\mathfrak{I}^{(k)}}
		&=&
		C_{\mathfrak{I}^{(1)}}\vartheta(C_{\mathfrak{I}^{(1)}})\,e^{-H_{-}/k}
		\, e^{-\vartheta(H_{-})/k}\cdots \nn
		&& \qquad \times \cdots
		C_{\mathfrak{I}^{(k)}}\vartheta(C_{\mathfrak{I}^{(k)}})\,e^{-H_{-}/k}\,e^{-\vartheta(H_{-})/k}\;.
	\eeq 
We assume in \eqref{eq:Hamiltonian_para} that $H_{-}\in\mathfrak{A}_{-}^{n}$. Thus $Y_{\mathfrak{I}^{(1)},\ldots,\,\mathfrak{I}^{(k)}}$  has the form in \eqref{Phase Factors-2} with $B_{j}=e^{-H_{-}/k}$ for all $j$. 
Let 
	\be
	D_{\mathfrak{I}^{(1)},\ldots,\mathfrak{I}^{(k)}}= C_{\mathfrak{I}^{(1)}}\,e^{-H_{-}/k}\,C_{\mathfrak{I}^{(2)}}\,
	e^{-H_{-}/k} \cdots C_{\mathfrak{I}^{(k)}}
	\, e^{-H_{-}/k} \in \mathfrak{A}_{-}\;.
	\ee
	
\begin{lem}[General Rearrangement]\label{Prop:Phase Factors}
Let   $C_{\mathfrak{I}^{(j)}}\in \mathfrak{A}_{-}$, and let $A, B_{j}\in \mathfrak{A}_{-}^{n}$, for $j=1,\ldots,k$.  Then  
	\beq\label{Phase Factors-2}
		&& \hskip -.3in A\vartheta(A)C_{\mathfrak{I}^{(1)}}\vartheta(C_{\mathfrak{I}^{(1)}})\,
		B_{1}\vartheta(B_{1}) \,
		C_{\mathfrak{I}^{(2)}}\vartheta(C_{\mathfrak{I}^{(2)}})\,
		B_{2}\vartheta(B_{2})
		\cdots  C_{\mathfrak{I}^{(k)}} \vartheta(C_{\mathfrak{I}^{(k)}})\,B_{k}\vartheta(B_{k})\nn
		&& \quad = \omega^{\sum_{1\leqslant j<j' \leqslant k}\vert\mathfrak{I}^{(j)}\vert\,
			 \vert\mathfrak{I}^{(j')}\vert}\   AD_{\mathfrak{I}_{1},\ldots,\mathfrak{I}_{k}} \,\vartheta(AD_{\mathfrak{I}_{1},\ldots,\mathfrak{I}_{k}})\;,
	\eeq
where 
	$
		D_{\mathfrak{I}^{(1)},\ldots,\mathfrak{I}^{(k)}}
		= C_{\mathfrak{I}^{(1)}}B_{1}\,C_{\mathfrak{I}^{(2)}}B_{2}
		 \cdots C_{\mathfrak{I}^{(k)}}B_{k} \in \mathfrak{A}_{-}\;,
	$
and correspondingly, 
	$
		\vartheta(D_{\mathfrak{I}^{(1)},\ldots,\mathfrak{I}^{(k)}})
		= \vartheta(C_{\mathfrak{I}^{(1)}})\vartheta(B_{1})
			\cdots \vartheta(C_{\mathfrak{I}^{(k)}})\vartheta(B_{k})
		 \in \mathfrak{A}_{+}\;.
	$
\end{lem}

\begin{proof}
In order to establish \eqref{Phase Factors-2},  rearrange the order of the factors on the left side of the identity.  In doing this, one retains the relative order of  $A$, of the various $C_{\mathfrak{I}^{(j)}}$, and of the various $B_{j'}$ that are elements of $\mathfrak{A}_{-}$.   Likewise one retains the relative order of $\vartheta(A)$, of the various $\vartheta(C_{\mathfrak{I}^{(j)}})$ and of the various $\vartheta(B_{j'})$ that are elements of $\mathfrak{A}_{+}$.  In this manner one obtains  $AD_{\mathfrak{I}^{(1)},\ldots,\mathfrak{I}^{(k)}}\vartheta(AD_{\mathfrak{I}^{(1)},\ldots,\mathfrak{I}^{(k)}})$ multiplied by some phase.

The resulting rearrangement only requires that one commutes 
operators in $\mathfrak{A}_{+}$ with operators in $\mathfrak{A}_{-}$.  As $\vartheta(A)\in\mathfrak{A}_{+}^{n}$ and $\vartheta(B_{j'})\in\mathfrak{A}_{+}^{n}$, each such factor commutes with every operator in $\mathfrak{A}_{-}$, and in particular with each  $C_{\mathfrak{I}^{(j)}}$.   Likewise $B_{j'}\in\mathfrak{A}_{-}^{n}$ commutes with each operator   $\vartheta(C_{\mathfrak{I}^{(j)}})$.   Thus one acquires a phase not equal to $1$, only  by moving one of the operators $\vartheta(C_{\mathfrak{I}^{(j)}})\in \mathfrak{A}_{+}$ to the right, past one of the operators $C_{\mathfrak{I}^{(j')}}\in \mathfrak{A}_{-}$.  And this is only required in case $j<j'$.  Use the rearrangement identity \eqref{eq:Reverse Theta C C} to perform this exchange.  This phase is given by the resulting product of phases arising in the elementary moves, and it yields the phase in   \eqref{Phase Factors-2}.    \hfill $\square$
\end{proof}

\begin{lem}[\bf Conservation Law]\label{lem:Counting_para} 
Let $A$ and $D_{\mathfrak{I}_{1},\ldots,\mathfrak{I}_{k}}$ be as in Lemma \ref{Prop:Phase Factors}.  Then the 
 trace of $AD_{\mathfrak{I}_{1},\ldots,\mathfrak{I}_{k}}\,\vartheta(AD_{\mathfrak{I}_{1},\ldots,\mathfrak{I}_{k}})$ vanishes  unless 
	\be\label{eq:condition_sigma para-2}
		\sum_{j=1}^{k} \vert \mathfrak{I}^{(j)}\vert=0\mod n\;.
	\ee  
If \eqref{eq:condition_sigma para-2} holds, then 
 the constants $\mathfrak{c}_{\mathfrak{I}^{(1)},\ldots,\,\mathfrak{I}^{(k)}}$ defined in \eqref{eq:product_couplings} satisfy
 	\be
		0\leqslant \mathfrak{c}_{\mathfrak{I}^{(1)},\ldots,\,\mathfrak{I}^{(k)}}\;.
	\ee  
\end{lem}

\begin{proof} 
Expand  $T=AD_ {\mathfrak{I}_{1},\ldots,\mathfrak{I}_{k}}$ and its reflection as a sum of monomials  \eqref{eq:A_expansion}, 
	\be
		T
		= \sum_{\widetilde{\mathfrak{I}}\subset\Lambda_{-}} a_{\widetilde{\mathfrak{I}}} C_{\widetilde{\mathfrak{I}}}\;,
		\quad\text{and}\quad
		\vartheta(T) 
		= \sum_{\widetilde{\mathfrak{I}'}\subset\Lambda_{-}} \overline{a_{\widetilde{\mathfrak{I}'}} }\,
		 \vartheta(C_{\widetilde{\mathfrak{I}'}})\;.		
	\ee
Here we distinguish $\widetilde{\mathfrak{I}}=\{ \widetilde{n}_{1}, \ldots, \widetilde{n}_{L/2},0,\ldots,0\}$ from $\mathfrak{I}^{(j)}$ in the definition of $C_{\mathfrak{I}^{(j)}}$.  Proposition \ref{prop:monomials_para} ensures that the trace of $C_{\widetilde{\mathfrak{I}}}\,\vartheta(C_{\widetilde{\mathfrak{I}}'})$ vanishes unless each $\widetilde{n}_{i}=0=\widetilde{n}'_{i}$.   The trace of $T\vartheta(T)$ is given by the constant term in the expansion in the monomial basis of parafermions.

Consider first the case in which  $A$ and all the $B_{j}$ are constants.  Then the relation \eqref{Reverse C} ensures that 
	\be
		T= C_{\mathfrak{I}^{(1)}}\cdots C_{\mathfrak{I}^{(k)}}
		= \alpha \,C_{\mathfrak{I}^{(1)}+\cdots+\mathfrak{I}^{(k)}}
		= \alpha\, C_{\widetilde{\mathfrak{I}}}\;,
		\quad\text{with}\quad \alpha\in\C\;,
	\ee
namely there is only one term $C_{\widetilde{\mathfrak{I}}}$ in the expansion of $T$.   Thus we have the local conservation law 
 	\be\label{eq:Local Conservation}
		\widetilde{n}_{i}
		= \sum_{j=1}^{k} n_{i}^{(j)}  \mod n \;,
	\ee
for each $i=1,\ldots,L$, and in fact $\widetilde{n}_{i}=0$ for $i>L/2$.

Proposition \ref{prop:monomials_para} ensures that the trace of $T\vartheta(T)$ vanishes unless each parafermion $c_{i}$ appears in $C_{\widetilde{\mathfrak{I}}}$  with an exponent equal to  $0$~mod~$n$.  In  other words $\widetilde{n}_{i}=0$.   Summing this relation  over $i$ gives the desired global conservation law \eqref{eq:condition_sigma para-2}.
 
 In the general case, the matrices $A$ and $B_{j}$ are elements of $\mathfrak{A}_{-}^{n}$.  One obtains $T$ from the previous case by replacing each $C_{\mathfrak{I}^{(j)}}$ by the product $C_{\mathfrak{I}^{(j)}}\,B_{j}$, and multiplying $D_ {\mathfrak{I}_{1},\ldots,\mathfrak{I}_{k}}$ by $A$.  One can expand $A$ and each $B_{j}$ using the basis of parafermion monomials, and the total degree of each non-zero term in each of these expansions is an integer multiple of $n$.  In the general case, the multiplications may introduce new parafermion factors,  so it may be the case that $\widetilde{n}_{i}\neq \sum_{j=1}^{k} n_{i}^{(j)}\mod n$, and the local conservation law \eqref{eq:Local Conservation} may not hold for $T$.  However the relation  \eqref{Reverse C}  ensures that each multiplication by $A$ or by $B_{j}$ changes the {\em total} degree of any monomial in the expansion of  $T$  by an integer multiple of $n$.  Thus
 	\be
		\sum_{i=1}^{L}
		\widetilde{n}_{i}
		 = \sum_{i=1}^{L}
		 \sum_{j=1}^{k} n_{i}^{(j)}\mod n
		 = \sum_{j=1}^{k} \vert \mathfrak{I}^{(j)}\vert \mod n\;,
	\ee
remains true.  Since  the trace of $T\vartheta(T)$ vanishes unless $\widetilde{n}_{i}=0$ for all $i$, we infer the global conservation law \eqref{eq:condition_sigma para-2}.  Hence \eqref{eq:condition_sigma para-2}  holds in the general case.

The positivity of  $\mathfrak{c}_{\mathfrak{I}^{(1)},\ldots,\,\mathfrak{I}^{(k)}}$   follows in case each of the coupling constants $J_{\mathfrak{I}^{(j)}\,\vartheta\mathfrak{I}^{(j)}}$ are non-negative.  In case of even $n$, we also allow a factor 
	\be 
	(-1)^{\sum_{j=1}^{k} \abs{\mathfrak{I}^{(j)}}}=(-1)^{\alpha n}
	\ee
for integer $\alpha$.   But as we are assuming that $n$ is even, this  also equals $+1$.  \hfill $\square$
\end{proof}

\noindent {\em Completion of the Proof of Theorem \ref{prop:reflection_positivity_para}.} 
Using \eqref{eq:Exp_minusH_k-2} and Lemma \ref{Prop:Phase Factors}, we infer that  
	\beqs\label{eq:Expansion_para}
		&&  \hskip -0.6cm A\,\vartheta(A)\lrp{e^{-H}}_{k}
		\nn 
		&& =\sum_{ \mathfrak{I}^{(1)},\ldots,\,\mathfrak{I}^{(k)}}
		(-1)^{\sum_{j=1}^{k}\vert\mathfrak{I}^{(j)}\vert}\,\omega^{\sum_{j=1}^{k}\frac{1}{2}\vert\mathfrak{I}^{(j)}\vert^{2}
		+ \sum_{1\leqslant j<j' \leqslant k}  \vert\mathfrak{I}^{(j)} \vert \,\vert \mathfrak{I}^{(j')}\vert}\,\nn
		&& \hskip 1.4in \times \
		\mathfrak{c}_{\mathfrak{I}^{(1)},\ldots,\,\mathfrak{I}^{(k)}} 
		AD_{\mathfrak{I}^{(1)},\ldots,\mathfrak{I}^{(k)}}\vartheta(AD_{\mathfrak{I}^{(1)},\ldots,\mathfrak{I}^{(k)}})\nn
		&& = 
		 \sum_{\mathfrak{I}^{(1)},\ldots,\,\mathfrak{I}^{(k)}}
		(-1)^{\sum_{j=1}^{k}\vert\mathfrak{I}^{(j)}\vert}\,\omega^{\sum_{j=1}^{k}\frac{1}{2}\vert\mathfrak{I}^{(j)}\vert^{2}
		+\frac{1}{2} \lrp{\sum_{j=1}^{k} \vert \mathfrak{I}^{(j)}\vert}^{2}
		-\frac{1}{2} \sum_{j=1}^{k}  \vert\mathfrak{I}^{(j)} \vert^{2} }\,\nn
		&& \hskip 1.4in \times \
		\mathfrak{c}_{\mathfrak{I}^{(1)},\ldots,\,\mathfrak{I}^{(k)}} 
		AD_{\mathfrak{I}^{(1)},\ldots,\mathfrak{I}^{(k)}}\,\vartheta(AD_{\mathfrak{I}^{(1)},\ldots,\mathfrak{I}^{(k)}})\nn
		&& = 
		 \sum_{\mathfrak{I}^{(1)},\ldots,\,\mathfrak{I}^{(k)}}
		(-1)^{\sum_{j=1}^{k}\vert\mathfrak{I}^{(j)}\vert}\,
		\omega^{\frac{1}{2}\lrp{\sum_{j=1}^{k}\vert\mathfrak{I}^{(j)}\vert}^{2}}\nn
		&& \hskip 1.4in \times \ 
		\mathfrak{c}_{\mathfrak{I}^{(1)},\ldots,\,\mathfrak{I}^{(k)}} 
		AD_{\mathfrak{I}^{(1)},\ldots,\mathfrak{I}^{(k)}}\,\vartheta(AD_{\mathfrak{I}^{(1)},\ldots,\mathfrak{I}^{(k)}})\;.
\eeqs
Taking the trace,
we have the approximation 
	\beq\label{eq:Good Approximation}
		&&\hskip -1.8cm \Tr\lrp{A\vartheta(A)\lrp{e^{-H}}_{k}}\nn
		&& \hskip -.3in 
		=  \sum_{\mathfrak{I}^{(1)},\ldots,\,\mathfrak{I}^{(k)}}
		(-1)^{\sum_{j=1}^{k}
		\vert\mathfrak{I}^{(j)}\vert}\,\omega^{\frac{1}
		{2}\lrp{\sum_{j=1}^{k}\vert\mathfrak{I}^{(j)}\vert}^{2}}\, 
		 \mathfrak{c}_{\mathfrak{I}^{(1)},\ldots,\,\mathfrak{I}^{(k)}} 
		\nn
		&&\hskip .7in
		\times \Tr\lrp{AD_{\mathfrak{I}^{(1)},\ldots,\mathfrak{I}^{(k)}}\,\vartheta(AD_{\mathfrak{I}^{(1)},\ldots,\mathfrak{I}^{(k)}})}\;.
	\eeq
From Lemma \ref{lem:Counting_para} we infer that the trace vanishes unless $\sum_{j=1}^{k} \vert \mathfrak{I}^{(j)}\vert=\alpha n$ for some 
non-negative integer $\alpha$.  Also in this case $\mathfrak{c}_{\mathfrak{I}^{(1)},\ldots,\,\mathfrak{I}^{(k)}} \geqslant 0 $.  The phase in \eqref{eq:Good Approximation} is 
	\[
		(-1)^{\sum_{j=1}^{k}
		\vert\mathfrak{I}^{(j)}\vert}\omega^{\frac{1}
		{2}\lrp{\sum_{j=1}^{k}\vert\mathfrak{I}^{(j)}\vert}^{2}}
		= (-1)^{\alpha n} \omega^{\frac12 \alpha^{2}n^{2}}	
		=  e^{ 2\pi i n \frac{\lrp{1+\alpha}\alpha}{2}}
		=1\;.
	\]
In the final equality we use the fact that $(1+\alpha)\alpha$ is even.
Proposition \ref{prop:positivity_para} ensures $\Tr(AD_{\mathfrak{I}^{(1)},\ldots,\mathfrak{I}^{(k)}}\vartheta(AD_{\mathfrak{I}^{(1)},\ldots,\mathfrak{I}^{(k)}})\geqslant0$.  So  each term in the sum \eqref{eq:Good Approximation} is non-negative.  Therefore the $k\to\infty$ limit of  \eqref{eq:Good Approximation}  is also non-negative.   \hfill $\square$
\end{proof}

\setcounter{equation}{0}
\section{RP Does Not Hold on  $\mathfrak{A_{-}}$}
We have proved that the functional $f(A)={\rm Tr}(A\,\vartheta(A)\,e^{-H})$ is  positive for  $A\in \mathfrak{A}^{n}_{-}\subset\mathfrak{A}_{-}$. This is what we defined as the algebra of observables  after \eqref{eq:Subalgebras}.  Here we remark that $f(A)$ is {\em not} positive on the full algebra $\mathfrak{A}_{-}$.  

Consider $L=2$ with the parafermion generators, $c=c_{1}\in  \mathfrak{A}_{-}^{1}$ and $c_{2}=\vartheta(c)^{*}\in \mathfrak{A}_{+}^{1}$.  Let $A=c$ and take $H=H_{0}=\omega^{\frac12}c\vartheta(c)$, which has the form \eqref{eq:Hamiltonian_para}--\eqref{eq:H0}, with $H_{-}=H_{+}=0$.  We now show that $f(c)$ is not positive, so $\vartheta$ is not RP on $\mathfrak{A}_{-}^{1}$. In fact 
\beqs
f(c)
&=& \sum_{k=0}^{\infty} \Tr(c\vartheta(c)(c\vartheta(c))^{k})\frac{(-1)^{k}\omega^{\frac{k}{2}}}{k!} \nn
&=&\sum_{k=0}^{\infty} \omega^{  \lrp{k  + \frac{k(k-1)}{2}}}\Tr\lrp{c^{1+k}\vartheta(c)^{1+k}}\frac{(-1)^{k}\omega^{\frac{k}{2}}}{k!} \nn
&=& \sum_{k=0}^{\infty} 
	\frac{\omega^{ \lrp{k+\frac{k^{2}}{2}}}  (-1)^{k}}
	{k!}
	\Tr\lrp{c^{1+k}\vartheta(c)^{1+k}} \;.
\eeqs
Use the fact that the trace vanishes unless $1+k=\ell n$ for $\ell=1,2,\ldots$.   Define $\sum^{(1)}_{k}$ as the sum over the subset of $k\in \mathbb{Z}_{+}$ for which $k=\ell n-1$ for some $\ell=\ell(k) \in \mathbb{Z}_{+}$.  Then   
\beq\label{RP not so}
f(c)
&=&  {\sum_{k}}^{(1)}\ 
\frac{\omega^{ \lrp{(\ell n-1)  + \frac{(\ell n-1)^{2}}{2}}}  (-1)^{\ell n-1}}
	{k!} \, \Tr(I)
\nn
&=&  - \omega^{- \frac12 }  {\sum_{k}}^{(1)}\ 
\frac{ \omega^{ \frac12 \ell^{2}n^{2}}  (-1)^{\ell n}}{k!}\,  \Tr(I)\;.
\eeq
For integer $\ell$, the product  $\ell(\ell+1)$  is an even, positive integer.  Thus the phase inside the sum equals  
	\be
	\omega^{ \frac12 \ell^{2}n^{2}}  (-1)^{\ell n}
	= e^{\lrp{\frac{2\pi i}{n}} \lrp{ \frac{1}{2}\ell^{2} n^{2}}+\pi i \ell n} 
	= e^{\pi i n\ell\lrp{\ell + 1}} = 1\;.
	\ee 
Therefore one finds that  
	\be
		f(c)
		= \omega^{\frac{n-1}{2}}\, \sinh 1\,\Tr(I)\not \in \mathbb{R}_{+}\;.
	\ee

One can also calculate  $f(c^{j})$ for the same Hamiltonian, noting that $c^{j}\in\mathfrak{A}^{j}_{-}$.  In this case  there are certain pairs $(n,j)$, with $j<n$, for which $f(c^{j})$ is positive.  Three such families of pairs are: 
	\begin{enumerate}
	\item{}  $n=k^{3}$, $j=k^{2}$, with $k\in\mathbb{Z}_{+}$,
	\item{}  $n=2k^{2}$, $j=2kj'$, with $1\leqslant j'<k$,
	\item{}   $n=k^{2}$, $j=j'k$ with $k$ odd and $1\leqslant j'<k$.  
	\end{enumerate}
We do not pursue the question of finding on exactly which subalgebras of $\mathfrak{A}_{-}$  the functional $f(c^{j})$ is positive.

\setcounter{equation}{0}
\section{The Baxter Clock Hamiltonian\label{Sect:Baxter Clock}}
As an example of a familiar parafermion interaction, Fendley has shown that the Baxter clock Hamiltonian (originally formulated as  interacting spins \cite{Baxter-1,Baxter-2}) can be expressed in terms of parafermions.  Near the end of \S3.2 of \cite{Fendley}, he finds that for parafermion generators $c_{j}$ of degree $n$,    
	\be\label{Baxter Hamiltonian-1}
		H 
		= \omega^{\frac{n-1}{2}}\sum_{j=1}^{L-1} t_{j} \,c_{j+1} c_{j}^{*}\;,
	\ee
where the $t_{j}$ are  real coupling constants.  As $c_{j}^{*}=c_{j}^{n-1}$, each term in the Hamiltonian is an element of the algebra $\mathfrak{A}^{n}$. 

In  \S\ref{sec:Introduction}  we remarked that if $\{c_{j}\}$ and parafermion generators, then $\{c_{j}^{*}\}$ are also parafermion generators. So using this alternative set of parafermions, one can also write the Baxter clock Hamiltonian as 
	\be\label{Baxter Hamiltonian-1a}
		H
		= \omega^{\frac{n-1}{2}}
		\sum_{j=1}^{L-1} t_{j} \,c_{j+1}^{*} c_{j}
		= - \omega^{\frac{1}{2}} 
		\sum_{j=1}^{L-1} t_{j} \,c_{j}\,c_{j+1}^{*}\;. 
	\ee 
One can split this sum into three parts,  
	\beq\label{Baxter Hamiltonian-1}
		H 
		&=&H_{-}+H_{0}+H_{+}\;,
	\eeq	
where  
	\[
		H_{-} 
		=  -\omega^{\frac{1}{2}} \sum_{j=1}^{\frac12 L-1} t_{j} \, c_{j}
		c_{j+1}^{*}\;,
		\quad
		H_{+} 
		=  -\omega^{\frac{1}{2}} \sum_{j=\frac12 L+1}^{ L-1} t_{j} \, c_{j}
		c_{j+1}^{*}\;, 
	\]
and 	
	\be
		H_{0} 
		= -\omega^{\frac12} \,t_{\frac12 L}\;
		 c_{\frac 12 L}
		\,c_{\frac12 L +1}^{*}
		= -\omega^{\frac12} \,t_{\frac12 L}\;
		 c_{\frac 12 L}
		\,\vartheta(c_{\frac12 L})\;.
	\ee
Note that $\vartheta(H_{0})=H_{0}$. Also 
	\beqs
		\vartheta(H_{-}) 
		&=& -\omega^{-\frac12}\sum_{j=1}^{\frac12 L-1} t_{j} \, 
			\vartheta(c_{j})	\vartheta (c_{j+1}^{*})
		= -\omega^{-\frac12}\sum_{j=1}^{\frac12 L-1} t_{j} \, 
			c_{L-j+1}^{*}	c_{L-j} \\
		&=& -\omega^{-\frac12-(n-1)} \sum_{j=1}^{\frac12 L-1} t_{j} \, 
				c_{L-j} c_{L-j+1}^{*}
		=-\omega^{\frac12} \sum_{j=\frac12 L+1}^{ L-1} t_{L-j} \, 
				c_{j} c_{j+1}^{*}\;.
	\eeqs

On the other hand, the parafermion Hamiltonians that we study in \eqref{eq:Hamiltonian_para}  include those with $\abs{\mathfrak{I}}=1$ of the form 
	\be\label{Case of PFHamiltonian-1}
		H 
		= H_{-}+H_{0}+H_{+}\;,
		\quad\text{with}\quad
		H_{+}=\vartheta(H_{-})\;,
	\ee  
and
	\beq\label{Case of PFHamiltonian-2}
		H_{0} 
		&=&   \omega^{\frac{1}{2}} J_{\frac12 L} \  c_{\frac12 L}\,
		\vartheta (c_{\frac12 L}) 
		 = \omega^{\frac{1}{2}} 
		 J_{\frac12 L} \  c_{\frac12 L}\,
			c_{\frac12 L+1}^{*} 
		\;.
	\eeq	

Thus Fendley's representation of the Baxter Hamiltonian has the required general form \eqref{Case of PFHamiltonian-1}--\eqref{Case of PFHamiltonian-2} if $J_{j}=-t_{j}$ for all $j$, and also  
	\be
		t_{L-j} = t_{j}	\;,
		\quad\text{for}\quad
		j=1,2,\ldots, \frac12 L-1\;.
	\ee   
Such a Hamiltonian is reflection invariant, $\vartheta(H)=H$, and it is gauge invariant $UHU^{*}=H$.  
It satisfies our RP hypotheses in \S\ref{Sect:Coupling Constants}  in case: 
\beq\label{CC Restriction}
&&\text{For odd } n{:}  \ t_{\frac12 L}\leqslant 0\;.\nn
&&\text{For even } n{:} \  t_{\frac12 L} \in \mathbb{R}\;.
\eeq

With periodic boundary conditions, when one wishes to place the reflection plane arbitrarily, one needs to require for RP that all the Baxter--Fendley coupling constants $\{t_{j}\}$ are equal, in addition to  \eqref{CC Restriction}.

\setcounter{equation}{0}  
\section{Reflection Bounds}\label{sect:Reflection_Bounds}
Reflection positivity allows one to define a pre-inner product on $\mathfrak{A}_{\pm}$ given by
\be\label{eq:inner_product_1}
\lra{A,B}=\Tr(A\,\vartheta(B))\,.
\ee
This pre-inner product satisfies the Schwarz inequality
\be\label{eq:Schwarz}
\abs{\lra{A,B}}^2\leqslant \lra{A,A}\,\lra{B,B}\,.
\ee
In the standard way, one obtains an inner product $\lra{\widehat{A},\widehat{B}}$ and norm $\Vert\widehat{A}\Vert$ by defining the inner product on equivalence classes $\widehat{A}=\{A+n\}$ of $A$'s, modulo elements $n$ of the null space of the functional \eqref{eq:inner_product_1} on the diagonal.   In order to simplify notation, we ignore this distinction.

Let us introduce two pre-inner products $\lra{\ \cdot\,,\cdot \ }_{\pm}$ on the algebras $\mathfrak{A}_{\pm}^{n}$, corresponding to two reflection-symmetric Hamiltonians. Let 
\be
\lra{A,B}_{-}=\Tr(A\,\vartheta(B)\,e^{-H_{-,\vartheta-}}),\quad\text{for}\quad H_{-,\vartheta -}=H_{-}+H_{0}+\vartheta(H_{-})\,.
\ee
Similarly define
\be
\lra{A,B}_{+}=\Tr(A\,\vartheta(B)\,e^{-H_{\vartheta+,+}})\,,\quad\text{for}\quad H_{\vartheta+,+}=\vartheta(H_{+})+H_{0}+H_{+}\,.
\ee
As previously, one can define inner products  and norms $\Vert \ \cdot\  \Vert_{\pm}$.

\begin{prop}[\bf RP-Bounds]\label{Prop:RPBound}
Let $H=H_{-}+H_{0}+H_{+}$ with $H_{\pm}\in\mathfrak{A}_{\pm}^{n}$ and $H_{0}$ of the form \eqref{eq:H0}. Then for $A,B\in\mathfrak{A}_{+}^{n}$,
\be\label{eq:bound_1}
\abs{\Tr(A\,\vartheta(B)\,e^{-H})}\leqslant \Vert A\Vert_{-}\,\Vert B\Vert_{+}\,.
\ee
Also
\be\label{eq:bound_2}
\abs{\Tr(A\,\vartheta(B)\,e^{-H})}\leqslant \Vert A\Vert_{+}\,\Vert B\Vert_{-}\,.
\ee
In particular for $A=B=I$, 
\be
\left\vert\Tr(e^{-H})\right\vert\leqslant \Tr(e^{-(H_{-}+H_{0}+\vartheta(H_{-}))})^{1/2}\,\Tr(e^{-(\vartheta(H_{+})+H_{0}+H_{+})})^{1/2}\,.
\ee
\end{prop}

\begin{proof}
The proof of \eqref{eq:bound_1} follows the proof of Theorem \ref{prop:reflection_positivity_para}.     \hfill $\square$
\end{proof}

\setcounter{equation}{0}
\section{Topological Order and Reflection Positivity}
In this section we impose periodic boundary conditions: allow the location label $i$ of the parafermion $c_{i}$ to take arbitrary integer values, and identifying the parafermion $c_{i}$ with $c_{j}$ when $i=j\mod L$.  Let  $W_{A}=A\,\vartheta(A)=\mathfrak{B}(C)$, be a loop of parafermions of length $2\ell$.  This means that $\mathfrak{B}(C)$ is a product of parafermion generators $O_{i}=c_{i}$,  
\begin{equation}\label{eq:WOi}
\mathfrak{B}(C)=O_{i_{1}}O_{i_{2}}\cdots O_{i_{2\ell}}\,,
\quad\text{where}\quad
i_{1}\leqslant i_{2} \leqslant \cdots \leqslant i_{2\ell}=i_{1}\;.
\end{equation}
(This choice is the most general, as $c_{i}^{n_{i}}$ is the product of several $c_{i}$'s.)  
Take $A=c_{i_{1}}\cdots c_{i_{\ell}}$ to be the product of parafermions  along half of a loop and $\vartheta(A)=c_{\vartheta i_{1}}^{*} \cdots c_{\vartheta i_{\ell}}^{*}=c_{2\ell}^{-1}\cdots c_{\ell+1}^{-1}$ the product of operators along the other half of the loop.

Consider a reflection-invariant Hamiltonian $H$, with a ground-state subspace $\mathcal{P}$.  Define $H$ to have $W$-order, if the operator $W$ applied to any vector $\Omega\in\mathcal{P}$ has no component in $\mathcal{P}$ that is orthogonal to $\Omega$.  In other words, $\mathcal{P}W\mathcal{P}$ is a scalar multiple of $\mathcal{P}$, and $W$ does not cause transitions between different ground states. Topological order involves the additional assumption that $W$ is localized.  

In an earlier paper \cite{TORP}, we have the following result for a Hamiltonian describing the interaction of Majoranas.   A similar argument shows that it applies as well to Hamiltonians describing the interaction of parafermions.  
\begin{prop}
Let $H$ be a reflection-positive Hamiltonian that has $W_{A}=A\vartheta(A)$  topological order, where $A\in\mathfrak{A}_{-}^{n}$. Then $0\leqslant\lra{\Omega, W_{A}\Omega}$ for any $\Omega\in\mathcal{P}$. 
\end{prop}

\setcounter{equation}{0}  \section{Acknowledgement}
We warmly thank Daniel Loss for suggesting that we learn about parafermions, and also for his hospitality at the University of Basel.  We acknowledge a useful presentation by Jelena Klinovaja. We are grateful to Maissam Barkeshli, Jacques Perk, and Thomas Quella for comments on an earlier version of this manuscript, including Perk's advising us of Sylvester's paper \cite{Sylvester-2}, and of an electronic archive where one can download \cite{Yamazaki}. We thank Barbara Drauschke for suggesting several improvements in style.

This work was supported by the Swiss NSF, NCCR QSIT, NCCR Nanoscience, and the Alexander von Humboldt Foundation.

\end{document}